\documentclass[twocolumn,english,conference]{IEEEtran}
\usepackage{mathptmx}
\usepackage[T1]{fontenc}
\usepackage[latin9]{inputenc}
\usepackage[paperwidth=8.5in,paperheight=11in]{geometry}
\usepackage{verbatim}
\usepackage{float}
\usepackage{amsmath}
\usepackage{graphicx}
\usepackage{setspace}
\usepackage{amssymb}

\makeatletter

\floatstyle{ruled}
\newfloat{algorithm}{tbp}{loa}
\floatname{algorithm}{Algorithm}

\newtheorem{example}{Example}
\newtheorem{remrk}{Remark}
\newtheorem{prop}{Proposition}
\newtheorem{lemma}{Lemma}
\newtheorem{thm}{Theorem}

\usepackage{wrapfig}
\usepackage{tikz}  
\usetikzlibrary{plotmarks}
\DeclareMathAlphabet{\mathcal}{OMS}{cmsy}{m}{n}

\IEEEoverridecommandlockouts

\makeatother

\usepackage{babel}

\begin{document}

\title{A Rate-Distortion Exponent Approach to Multiple Decoding Attempts
for Reed-Solomon Codes}

\author{Phong S. Nguyen, Henry D. Pfister, and Krishna R. Narayanan%
\thanks{This material is based upon work supported by the National Science
Foundation under Grant No. 0802124. The work of P. Nguyen was also
supported in part by a Vietnam Education Foundation fellowship. Any
opinions, findings, conclusions, or recommendations expressed in this
material are those of the authors and do not necessarily reflect the
views of the National Science Foundation.%
}\\
{\normalsize Department of Electrical and Computer Engineering,
Texas A\&M University }\\
{\normalsize College Station, TX 77840, U.S.A. }\\
{\normalsize \{psn, hpfister, krn\}@tamu.edu }\vspace{-2mm}}
\maketitle
\begin{abstract}
Algorithms based on multiple decoding attempts of Reed-Solomon (RS)
codes have recently attracted new attention. Choosing decoding candidates
based on rate-distortion theory, as proposed previously by the authors,
currently provides the best performance-versus-complexity trade-off.
In this paper, an analysis based on the rate-distortion exponent is
used to directly minimize the exponential decay rate of the error
probability. This enables rigorous bounds on the error probability
for finite-length RS codes and leads to modest performance gains.
As a byproduct, a numerical method is derived that computes the rate-distortion
exponent for independent non-identical sources. Analytical results
are given for errors/erasures decoding.
\end{abstract}
\thispagestyle{empty}\pagestyle{empty}

\section{Introduction}

%
\begin{comment}
On the coding front, Reed-Solomon (RS) codes have been the choice
for a wide range of applications in digital communication and data
storage systems. In most of the cases, the classical hard-decision
(HD) decoding algorithm of Berlekamp-Massey (BM) with low complexity
has been employed. However, HD decoding algorithms do not take advantage
of the soft-information available at the receiver. The potential gain
from making use of such information is very significant as exemplified
by the improvement of the soft-decision Koetter-Vardy (KV) algorithm
\cite{Koetter-it03} over the hard-decision Guruswami-Sudan (GS) counterpart
\cite{Guruswami-it99}. One of the current interests of researchers
is to find soft-decision decoding methods for RS codes with improved
performance at a reasonable expense of complexity.
\end{comment}
{}The design of a computationally efficient soft-decision decoding algorithm
for Reed-Solomon (RS) codes has been the topic of significant research
interest for the past several years. Currently, there are several
soft-decision decoding algorithms for RS codes which exhibit a wide
range of trade-offs between computational complexity and error performance. 

Among such decoding methods is a class of algorithms called multiple
errors-and-erasures decoding. The algorithms belonging to this class
first construct a set of erasure patterns based on the available soft
information and then run an errors-and-erasures decoding algorithm,
such as the Berlekamp-Massey (BM) algorithm, multiple times. Each
time one erasure pattern in the set is used for decoding. By doing
this, the algorithm outputs a list of candidate codewords and then
chooses the best codeword from the list. Several algorithms of this
type, including the popular generalized minimum distance (GMD) decoding
algorithm, are discussed in \cite{Forney-it66,Lee-globecom08,Bellorado-isit06,Nguyen-aller09}. 

In \cite{Nguyen-aller09}, the authors proposed a rate-distortion
(RD) approach for constructing the set of erasure patterns. The main
idea is to choose an appropriate distortion measure so that the decoding
is successful if and only if the distortion between the error pattern
and erasure pattern is smaller than a fixed threshold. After that,
a set of erasure patterns is generated randomly (similar to a random
codebook generation) in order to minimize the expected minimum distortion.
The approach was also extended to analyze multiple-decoding for decoding
schemes beyond conventional errors-and-erasures decoding. 

\begin{singlespace}
One of the drawbacks in the RD approach is that the mathematical framework
is only valid as the block-length goes to infinity. Therefore, we
also consider the natural extension to a rate-distortion \emph{exponent}
(RDE) approach that studies the behavior of the probability, $p_{e}$,
that the transmitted codeword is not on the list as a function of
the block-length $N$. The overall error probability can be approximated
by $p_{e}$ because the probability that the transmitted codeword
is on the list but not chosen is very small compared to $p_{e}$.
Hence, our new approach essentially focuses on investigating the exponent
at which the error probability decays as $N$ goes to infinity. 
\end{singlespace}

The proposed RDE approach can also be considered as the generalization
of the RD approach since the latter is a special case of the former
when the RDE function tends to zero. Using the RDE analysis, our proposed
approach also helps answer the following two questions: (i) What is
the maximum rate-distortion exponent achievable at or below a given
number of decoding attempts (or a given size of the set of erasure
patterns)? (ii) What is the minimum number of decoding attempts required
to achieve a rate-distortion exponent at or above a given level? 

The paper is organized as follows. In Section \ref{sec:MultipleBMA},
we review multiple errors-and-erasures decoding algorithms and highlight
the connection between multiple errors-and-erasures decoding and rate-distortion.
Then, in Section \ref{sec:EXP-approach}, we propose a RDE approach
to construct a good set of erasure patterns for a finite length codewords.
Next, we discuss how to compute the RDE function which is required
in the proposed approach. Finally, simulation results are presented
in Section \ref{sec:Simulation} and conclusion is provided in Section
\ref{sec:Conclusion}.

\section{Multiple errors-and-erasures Decoding\label{sec:MultipleBMA}}

In this section, we discuss several multiple errors-and-erasures decoding
algorithms. While each algorithm uses a different set of erasure patterns,
the common trend is that one either erases or tries several different
candidates for each symbol in the least reliable positions (LRPs).
One focuses on the LRPs because the hard-decision made at these positions
are more likely to be incorrect. 

Let $\mathbb{F}_{m}$ be the Galois field with $m$ elements denoted
as $\alpha_{1},\alpha_{2},\ldots,\alpha_{m}$. We consider an $(N,K$)
RS code of length $N$ and dimension $K$ over $\mathbb{F}_{m}$.
Assume that we send a codeword $\mathbf{c}=(c_{1},c_{2},\ldots,c_{N})$
over some channel and $\mathbf{r}=(r_{1},r_{2},\ldots,r_{N})$ is
the received vector. A well-known decoding threshold states that a
single attempt of errors-and-erasures decoding succeeds if and only
if\vspace{-2mm} \begin{equation}
2v+e<d_{min}=N-K+1\label{eq:decthreshold}\end{equation}

\vspace{-6mm}~

\hspace{-3.5mm}where $e$ is number of erased symbols and $v$ is
the number of errors in unerased positions. A multiple errors-and-erasures
decoding is considered to succeed if the decoding threshold (\ref{eq:decthreshold})
is satisfied for at least one attempt of decoding. Intuitively, the
best case is when one erases an error and the worst case is when ones
wastes an erasure on a hard-decision symbol that turns out be correct.

The first algorithm of this type is called Generalized Minimum Distance
(GMD) decoding \cite{Forney-it66} where the set of erasure patterns
is obtained by successively erasing the $0,2,4,\ldots,d_{min}-1$
LRPs (with the assumption that the minimum distance $d_{min}$ is
odd). Recent work by Lee and Kumar \cite{Lee-globecom08} proposes
a soft-information successive (multiple) error-and erasure decoding
(SED) which constructs the set of erasure patterns in a more exhaustive
way. Specifically, SED$(l,f)$ tries to erase all possible combinations
of an even number less than or equal to $f$ of positions within the
$l$ LRPs. The SED algorithm hence yields better performance but at
increased complexity. 

In an attempt to answer the question how to build a good set of erasure
patterns in terms of performance-versus-complexity, in \cite{Nguyen-aller09},
we proposed a probabilistic method based on rate-distortion theory
and random coding arguments instead of the deterministic methods which
had been used in previously proposed algorithms. Basically, after
defining $x^{N}$ and $\hat{x}^{N}$ as an error pattern and an erasure
pattern whose letters $x_{i}$'s and $\hat{x}_{i}$'s are in the alphabets
$\mathcal{X}$ and $\hat{\mathcal{X}}$ respectively, a letter-by-letter
distortion measure $\delta:\mathcal{X}\times\hat{\mathcal{X}}\to\mathbb{R}_{\geq0}$
is chosen properly so that the condition (\ref{eq:decthreshold})
can be reduced to the form\vspace{-1mm}\begin{equation}
d\left(x^{N},\hat{x}^{N}\right)<N-K+1\label{eq:distorthreshold}\end{equation}

\vspace{-6mm}~

\hspace{-3.5mm}where the distortion between an error pattern and
an erasure pattern $d\left(x^{N},\hat{x}^{N}\right)=\sum_{i=1}^{N}\delta(x,\hat{x}_{i})$
is smaller than a fixed threshold. In general, an appropriate distortion
measure $\delta(j,k)$ for every $j\in\mathcal{X}$ and $k\in\hat{\mathcal{X}}$
should be specified.
\begin{example}
Consider a specific class of multiple errors-and-erasures (Berlekamp-Massey)
top-$\ell$ decoding (mBM-$\ell$) for an positive integer $\ell$
smaller than the field size $m$ where at each codeword index, up
to the $\ell$-th most likely symbols are taken care of. In this case,
$\mathcal{X}=\hat{\mathcal{X}}=\mathbb{Z}_{l+1}$ and $x^{N}\in\mathcal{X}^{N}$
where at each index $i$, $x_{i}=0$ implies that using up to the
$\ell$-th most likely symbols as the hard-decision all gives an error,
$x_{i}=j$ implies that the $j$-th most likely symbol is correct
for $j=1,2,\ldots,\ell$; $\hat{x}^{N}\in\hat{\mathcal{X}}^{N}$ where
at each index $i$, $\hat{x}_{i}=0$ implies that an erasure is applied
and $\hat{x}_{i}=k$ implies that the $k$-th most likely symbol is
used as the hard-decision for $k=1,2,\ldots,\ell$. For example, mBM-1
is the case of multiple conventional errors-and-erasures decoding.
The letter-by-letter distortion measure for mBM-1 is chosen in the
following way\vspace{-1mm}\begin{equation}
\begin{array}{cc}
\delta(0,0)=1 & \delta(0,1)=2\\
\delta(1,0)=1 & \delta(1,1)=0.\end{array}\label{eq:dstmb1}\end{equation}
\vspace{-4mm}

It is also possible to choose appropriate distortion measures that
work for $\ell>1$ and other decoding schemes such as algebraic soft-decision
(ASD) decoding. Still, the main idea is to convert the decoding threshold
of the corresponding decoding scheme into the form of (\ref{eq:distorthreshold}).

Thus, by viewing $x^{N}$ (resp. $\hat{x}^{N}$) as a source sequence
(resp. reproduction sequence) and choosing a suitable distortion measure,
the task of designing a good set of erasure patterns turns out to
be how to best approximate the source sequence with a minimum number
of reproduction sequences. In other words, it becomes a covering problem
where one wants to cover the most-likely error patterns with the fewest
number of balls whose centers are erasure patterns. The main steps
in the RD based algorithm are given here briefly, but more detail
can be found in \cite{Nguyen-aller09}.

\emph{Step 1}: Empirically compute the reliability matrix whose entries
are $\Pr(c_{i}=\alpha_{j}|r_{i})$ for $i=1,2,\ldots N$ and $j=1,2,\ldots m$.
From this, determine probability matrix $\mathbf{P}$ where $p_{i,j}=\Pr(x_{i}=j$)
for $i=1,2,\ldots,N$ and $j\in\mathcal{X}$.

\emph{Step 2}: Compute the RD function using $\mathbf{P}$. Determine
the test-channel input-distribution matrix $\mathbf{Q}$ where $q_{i,k}=\mbox{\ensuremath{\Pr}}(\hat{x}_{i}=k)$
for $i=1,2\ldots,N$ and $k\in\mathcal{\hat{X}}$ that achieves a
point on the RD curve corresponding to a chosen rate $R$.

\emph{Step~3: }Randomly generate a set $\mathcal{B}$ of~$2^{R}$~erasure
patterns using the distribution matrix $\mathbf{Q}$ in the correct
reliability order of the codeword positions.

\emph{Step~4:} Run multiple attempts of the corresponding decoding
scheme using the set $\mathcal{B}$ to produce a list of candidate
codewords.

\emph{Step~5:} Use Maximum-Likelihood (ML) decoding to pick the best
codeword on the list.\vspace{-2mm}
\end{example}

\section{RATE-DISTORTION EXPONENT APPROACH\label{sec:EXP-approach}}

In the RD approach, the set $\mathcal{B}$ of $2^{N\bar{R}}$ (or
$2^{R}$) erasure patterns can be generated randomly so that%
\footnote{We denote the rate and distortion by $R$ and $D$, respectively,
using unnormalized quantities, i.e., $R=N\bar{R}$ and $D=N\bar{D}$.%
}\vspace{-2mm} \[
\lim_{N\rightarrow\infty}\frac{1}{N}E_{x^{N},\mathcal{B}}[\min_{\hat{x}^{N}\in\mathcal{B}}d(x^{N},\hat{x}^{N})]<\bar{D}.\]

\vspace{-4mm}~

Thus, for large enough $N$, with high probability we have $\min_{\hat{x}^{N}\in\mathcal{B}}d(x^{N},\hat{x}^{N})<N\bar{D}=D$.
Basically, \cite{Nguyen-aller09} focuses on minimizing the average
minimum distortion with little knowledge of how the tail of the distribution
behaves. In this paper, we instead focus on directly minimizing the
probability that the minimum distortion is not less than the pre-determined
threshold $D=N-K+1$ (due to the condition (\ref{eq:distorthreshold}))
with the help of an error-exponent analysis. The exact probability
of interest is $p_{e}=\Pr(x^{N}:\min_{\hat{x}^{N}\in\mathcal{B}}d(x^{N},\hat{x}^{N})>D)$
that reflects how likely the decoding threshold (\ref{eq:decthreshold})
is going to fail.

In other words, every error pattern $x^{N}$ can be covered by a sphere
centered at an erasure pattern $\hat{x}^{N}$ except for a set of
error patterns of probability $p_{e}$. The RDE analysis shows that
$p_{e}$ decays exponentially as $N\to\infty$ and the maximum exponent
attainable is the RDE function. In our context, we have a source sequence
$x^{N}$ of $N$ independent non-identical source components. We denote
the rate-distortion exponent by $F(R,D)$ using unnormalized quantities
(i.e., without dividing by $N$) and note that exponent used by other
authors in \cite{Blahut-it74,Marton-it74,Csiszar-1981} is often the
normalized version $\bar{F}(R,D)\triangleq\frac{F(R,D)}{N}$. 

The original RDE function $F(R,D)$, defined in \cite{Blahut-it74}
for a single source $x$, is given by%
\footnote{All logarithms are taken to base 2.%
}\vspace{-2mm}\[
F(R,D)=\max_{\mathbf{w}}\min_{\mathbf{\tilde{p}}\in\mathcal{P}_{R,D}}\sum_{j}\tilde{p}_{j}\log\frac{\tilde{p}_{j}}{p_{j}}\]
\vspace{-3mm}~

\hspace{-3.5mm}where $p_{j}\triangleq\Pr(x=j)$, $w_{k|j}\triangleq\Pr(\hat{x}=k|x=j)$,
and\vspace{-2mm}

\[
\mathcal{P}_{R,D}=\left\{ \mathbf{\tilde{p}}\bigg|\begin{array}{c}
\sum_{j}\sum_{k}\tilde{p}_{j}w_{k|j}\log\frac{w_{k|j}}{\sum_{j}\tilde{p}_{j}Q_{k|j}}\geq R\\
\sum_{j}\sum_{k}\tilde{p}_{j}w_{k|j}\delta_{jk}\geq D\end{array}\right\} .\]
\vspace{-2mm}%
\begin{comment}
\[ \begin{array}{ccc} \mathcal{P}_{R,D} & = & \Big\lbrace\tilde{p}:  \sum_{j}\sum_{k}\tilde{p}_{j}w_{k|j}\log\frac{w_{k|j}}{\sum_{j}\tilde{p}_{j}Q_{k|j}}\geq R\\  &   & \sum_{j}\sum_{k}\tilde{p}_{j}w_{k|j}\delta_{jk}\geq D\Big\rbrace.\end{array}\] 
\end{comment}
{}

The RDE was first extensively discussed in \cite{Blahut-it74,Marton-it74}
and their results show that there exists a set $\mathcal{B}$ of roughly
$2^{N\bar{R}}$ codewords, generated randomly using the test-channel
input distribution matrix $\mathbf{Q}$, that achieves $\bar{F}(R,D)$.
This gives the upper bound that for every $\epsilon>0$, we have\vspace{-1mm}
\begin{equation}
p_{e}\leq2^{-N[\bar{F}(R,D)-\epsilon]}.\label{eq:BlahutUB}\end{equation}

\vspace{-6.5mm}~

\hspace{-3.5mm}for $N$ large enough (see \cite[p. 229]{Blahut-1987}).
An exponentially tight lower bound for $p_{e}$ can also be obtained
for $N$ large enough (see \cite[p. 236]{Blahut-1987}) and this gives\vspace{-2mm}\[
\lim_{N\to\infty}-\frac{1}{N}\log p_{e}=\bar{F}(R,D).\]
\vspace{-5mm}~

\hspace{-3.5mm}\noindent\hspace{1em}{\itshape Proposed algorithm:}
In the RDE approach proposed here, instead of computing the RD function,
we need to compute the RDE function $F(R,D)$ along with the optimal
test-channel input distribution matrix $\mathbf{Q}$ (see Section
\ref{sec:EXP-func}). This distribution serves as a replacement for
the distribution used in Step 2 of the RD based algorithm given in
the previous section. Apart from this, the other steps of that algorithm
are unchanged for the proposed RDE-based algorithm. 
\begin{remrk}
\label{rem:The-RDE-approach}The RDE approach possesses several advantages.
First, it can help one estimate the smallest number of decoding attempts
to get to a RDE of $F$ (or get to an error probability of roughly
$2^{-N\bar{F}}$) or, similarly, allow one to estimate the RDE (and
error probability) for a fixed number of decoding attempts. Second,
it provides a converse based on the sphere-packing bound lower bound
for $p_{e}$. This implies that, given an arbitrary set $\mathcal{B}$
of roughly $2^{N\bar{R}}$ erasure patterns and any $\epsilon>0$,
the probability $p_{e}$ cannot be made lower than $2^{-N[\bar{F}(R,D)+\epsilon]}$
for $N$ large enough. Thus, no matter how one chooses the set $\mathcal{B}$
of erasure patterns, the difference between the induced probability
of error and the $p_{e}$ for the RDE approach becomes negligible
for $N$ large enough.\vspace{-4.3mm}
\end{remrk}
~
\begin{remrk}
It is interesting to note that the RDE approach for ASD decoding schemes
contains the special case where the codebook has only one entry (i.e.,
ASD decoding is run one time). In this case, the distribution of the
codebook that maximizes the exponent implicitly generates the optimal
multiplicity matrix. This is similar to the line of work \cite{El-Khamy-dimacs05,Ratnakar-it05,Das-isit09}
where various researchers tried to find the multiplicity matrix that
optimizes the error-exponent obtained by either applying a Chernoff
bound \cite{El-Khamy-dimacs05,Ratnakar-it05} or using Sanov's theorem
\cite{Das-isit09}.\vspace{-2mm}
\end{remrk}

\section{COMPUTING THE RDE FUNCTION\label{sec:EXP-func}}

In this section, we first present an extension of Arimoto's numerical
method for computing the RDE function \cite{Arimoto-it76} that works
for any chosen discrete distortion measure. Next, we consider some
special case where we can give an analytical treatment of the function.\vspace{-1mm}

\subsection{Numerical computation of RDE function}

For each discrete source component $x_{i}$, given two parameters
$s\geq0$ and $t\leq0$, the Arimoto algorithm given in \cite{Arimoto-it76}
computes the RDE function numerically as follows.

$\bullet$ Step 1: Choose an arbitrary all-positive distribution vector~$\underline{q}^{(0)}=\left(q_{1}^{(0)},q_{2}^{(0)},\ldots,q_{|\hat{\mathcal{X}|}}^{(0)}\right)$.

$\bullet$ Step 2: Iterate the following steps at $\tau=0,1,\ldots$
\vspace{-4mm}

\begin{eqnarray*}
w_{k|j}^{(\tau)} & = & \frac{q_{k}^{(\tau)}2^{t\delta_{jk}}}{\sum_{k}q_{k}^{(\tau)}2^{t\delta_{jk}}}\\
q_{k}^{(\tau+1)} & = & \frac{\left\{ \sum_{j}p_{j}2^{-st\delta_{jk}}(w_{k|j}^{(\tau)})^{(1+s)}\right\} ^{\frac{1}{1+s}}}{\sum_{k}\left\{ \sum_{j}p_{j}2^{-st\delta_{jk}}(w_{k|j}^{(\tau)})^{(1+s)}\right\} ^{\frac{1}{1+s}}}.\end{eqnarray*}
\vspace{-3mm}~

\hspace{-3.5mm}for $j\in\mathcal{X}$ and $k\in\hat{\mathcal{X}}$.

It is shown by Arimoto that $w_{k|j}^{(\tau)}\rightarrow w_{k|j}^{\star}$
and $q_{k}^{(\tau)}\rightarrow q_{k}^{\star}$ as $\tau\rightarrow\infty$.
Using the resulting $w_{k|j}^{\star}$ and $q_{k}^{\star}$, we can
compute\vspace{-2mm}\begin{eqnarray}
F & = & \sum_{j}\tilde{p}_{j}^{\star}\log\frac{\tilde{p}_{j}^{\star}}{p_{j}}\vspace{-4mm}\label{eq:Ffromq}\\
R & = & \sum_{j}\sum_{k}\tilde{p}_{j}^{\star}w_{k|j}^{\star}\log\frac{w_{k|j}^{\star}}{\sum_{j}\tilde{p}_{j}^{\star}w_{k|j}^{\star}}\vspace{-2mm}\label{eq:Rfromq}\\
D & = & \sum_{j}\sum_{k}\tilde{p}_{j}^{\star}w_{k|j}^{\star}\delta_{jk}\vspace{-2mm}\label{eq:Dfromq}\end{eqnarray}
\vspace{-2mm}~

\hspace{-3.5mm}where $\tilde{p}_{j}^{\star}=\frac{p_{j}(\sum_{k}q_{k}^{\star}2^{t\delta_{jk}})^{-s}}{\sum_{j}p_{j}(\sum_{k}q_{k}^{\star}2^{t\delta_{jk}})^{-s}}$.
\vspace{1mm}

However, in the context we consider, the source (error pattern) $x^{N}$
comprises independent but not necessarily identical source components
$x_{i}$'s. The complexity is a problem if we consider a group of
source letters $(j_{1},j_{2},\ldots,j_{N})$ as a supper-source letter
$\mathcal{J}$, a group of reproduction letters $(k_{1},k_{2},\ldots,k_{N})$
as a super-reproduction letter $\mathcal{K}$ and apply the Arimoto
algorithm straightforwardly . Instead, we can avoid this computational
obstacle by choosing the initial distribution still arbitrarily but
following a factorization rule $q_{\mathcal{K}}^{(0)}=\prod_{i=1}^{N}q_{k_{i}}^{(0)}$.
Then, we can verify that this factorization rule still holds for $w_{\mathcal{K}|\mathcal{J}}^{(\tau)}$
and $q_{\mathcal{K}}^{(\tau)}$ after every step of the Arimoto algorithm.
This leads to\vspace{-1mm} \[
\begin{array}{ccc}
w_{\mathcal{K}|\mathcal{J}}^{\star}=\prod_{i=1}^{N}w_{k_{i}|j_{i}}^{\star} & \mbox{and} & q_{\mathcal{J}}^{\star}=\prod_{i=1}^{N}q_{k_{i}}^{\star}.\end{array}\]

Combining with $\delta_{\mathcal{JK}}=\sum_{i=1}^{N}\delta_{j_{i}k_{i}}$
and $p_{\mathcal{J}}=\prod_{i=1}^{N}p_{j_{i}},$ we have\vspace{-2mm}
\[
\tilde{p}_{\mathcal{J}}^{\star}=\prod_{i=1}^{N}\tilde{p}_{j_{i}}^{\star}.\]
\vspace{-3mm}

This gives the following proposition. 
\begin{prop}
(Factored Arimoto algorithm for RDE function) Consider a discrete
source $x^{N}$ of independent but non-identical source components
$x_{i}$'s. Given parameters $s\geq0$ and $t\leq0$, the exponent,
rate and distortion are given by\vspace{-2mm}\[ \left.F\right|_{s,t}=\sum_{i=1}^{N}\left.F_i\right|_{s,t},\,\,\, \left.R\right|_{s,t}=\sum_{i=1}^{N}\left.R_i\right|_{s,t},\,\,\, \left.D\right|_{s,t}=\sum_{i=1}^{N}\left.D_i\right|_{s,t}\]where
the components $\left.F_i\right|_{s,t},\left.R_i\right|_{s,t},\left.D_i\right|_{s,t}$
are computed parametrically by the Arimoto algorithm.\vspace{-1mm}
\end{prop}

\subsection{Analytical computation of RDE function}

In this subsection, we consider the case m-BM1 whose distortion measure
is given in (\ref{eq:dstmb1}). We study the setup that RS codewords
defined over Galois field $\mathbb{F}_{m}$ are transmitted over the
$m$-ary symmetric channel ($m$-SC) which for each parameter $p$
can be modeled as\[
\Pr(r|c)=\begin{cases}
p & \mbox{if}\,\, r=c\\
(1-p)/(m-1) & \mbox{if}\,\, r\neq c.\end{cases}\]
Here, $c$ (resp. $r$) is the transmitted (resp. received) symbol
and $r,c\in\mathbb{F}_{m}$. With this channel model, we consider
$p$ not too small so that $p>(1-p)/(m-1)$. Therefore, at each index
$i$ of the codeword, the hard-decision is also the received symbol
and then it is correct with probability $p$. Thus, we have $p_{i,1}\triangleq\Pr(x_{i}=1)=p$
for every index $i$ of the error pattern $x^{N}$. That means, in
this context we have a source $x^{N}$ with i.i.d. binary components
$x_{i}$. Since the components $x_{i}$ are i.i.d we can treat each
$x_{i}$ as a binary source $X$ with $\Pr(X=1)\triangleq p$ and
$\Pr(X=0)=1-p\triangleq\bar{p}$ and first compute the RDE function
for this source $X$.

According to \cite{Blahut-it74}, for any admissible $(R,D)$ pair
we can find two parameters $s\geq0$ and $t\leq0$ so that $F(R,D)$
can be parametrically evaluated as\begin{eqnarray*}
F(R,D) & = & sR-stD+\max_{q_{1}}\left(-\log f(q_{1})\right)\\
 & = & sR-stD-\log\min_{q_{1}}f(q_{1})\end{eqnarray*}
\vspace{-4mm}where\[
f(q_{1})=\bar{p}\left(\sum_{k}q_{k}2^{t\delta_{0k}}\right)^{-s}+p\left(\sum_{k}q_{k}2^{t\delta_{1k}}\right)^{-s}\]
and $R,D$ are given in terms of optimizing $q^{\star}$ which we
will discuss later.

For the distortion measure in (\ref{eq:dstmb1}) and note that $q_{0}=1-q_{1}$,
we have\vspace{-1mm} \[
f(q_{1})=\bar{p}\left((1-q_{1})2^{t}+q_{1}2^{2t}\right)^{-s}+p\left((1-q_{1})2^{t}+q_{1}\right)^{-s}\]
which is a convex function in $q_{1}$. We then see that\vspace{-2mm}\[
\frac{\partial f}{\partial q_{1}}=0\Leftrightarrow q_{1}^{\star}=\frac{1+2^{t}}{1-2^{t}}\left(\frac{1}{1+2^{t}}-\frac{\bar{p}^{\frac{1}{s+1}}}{2^{\frac{st}{s+1}}p^{\frac{1}{s+1}}+\bar{p}^{\frac{1}{s+1}}}\right)\triangleq\beta.\]

In order to minimize $f(q_{1})$ over $q_{1}\in[0,1]$, we consider
three following cases where the optimal $q_{1}^{\star}$ is either
on the boundary or at a point with zero gradient.

$\bullet$ Case 1: $0\leq p\leq\frac{2^{t}}{1+2^{t}}$ then $\beta\leq0$.
Since $f$ convex, it is non-decreasing in the interval $[\beta,\infty)$
and therefore in the interval $[0,1]$. Thus, the optimal $q_{1}^{\star}=0$
and we can also compute from (\ref{eq:Ffromq}), (\ref{eq:Rfromq}),
(\ref{eq:Dfromq}) that\vspace{-2mm}\[
\begin{array}{ccc}
D=1; & R=0; & F=0=D_{KL}(u||p)\end{array}\]
\vspace{-7mm}~

\hspace{-3.5mm}where in this case $u=p$.

$\bullet$ Case 2: $1\geq p\geq\frac{1}{1+2^{t(2s+1)}}$ then $\beta\geq1$.
Since $f$ convex, it is non-increasing in the interval $(-\infty,\beta]$
and therefore in the interval $[0,1]$. Thus, the optimal $q_{1}^{\star}=1$
and similarly we get\vspace{-2mm}\[
\begin{array}{ccc}
D=\frac{2\bar{p}}{p2^{2ts}+\bar{p}}; & R=0; & F=D_{KL}(u||p)\end{array}\]
%
\begin{comment}
\begin{eqnarray*}
D & = & \frac{2\bar{p}}{p2^{2ts}+\bar{p}};R=0;F\\
R & = & 0\\
F & = & I(u||p)\end{eqnarray*}

\end{comment}
{}where in this case $u=1-\frac{D}{2}$. We can further see that $D\in[2(1-p),1]$
and $u\in[1-D,p]$. 

$\bullet$ Case 3: $\frac{2^{t}}{1+2^{t}}<p<\frac{1}{1+2^{t(2s+1)}}$
then $\beta\in(0,1)$. In this case, the optimal $q_{1}^{\star}=\beta$.
We then can find $w_{k|j}^{\star}=\frac{q_{k}^{\star}2^{t\delta_{jk}}}{\sum_{k}q_{k}^{\star}2^{t\delta_{jk}}}$
according to \cite{Blahut-it74} and plug in (\ref{eq:Ffromq}), (\ref{eq:Rfromq}),
(\ref{eq:Dfromq}) to get%
\footnote{The binary entropy function is $H(u)\triangleq-u\log u-(1-u)\log(1-u)$
and the Kullback-Leibler divergence is $D_{KL}(u||p)\triangleq u\log\frac{u}{p}+(1-u)\log\frac{1-u}{1-p}$.%
}\vspace{-2mm}\begin{eqnarray*}
D & = & \frac{2^{t}}{1+2^{t}}+1-u\\
R & = & H(u)-H(u+D-1)\\
F & = & D_{KL}(u||p)\end{eqnarray*}
where $u=\frac{2^{\frac{st}{s+1}}p^{\frac{1}{s+1}}}{2^{\frac{st}{s+1}}p^{\frac{1}{s+1}}+\bar{p}^{\frac{1}{s+1}}}$
. With this notation of $u$, we can express $q_{1}^{\star}=\frac{1-D}{3-2(u+D)}$
and $q_{0}^{\star}=\frac{2(1-u)-D}{3-2(u+D)}.$ We can see that $D\in(1-p,1)$.
It can also be verified that, in this case, by varying $s$ and $t,$
$u$ spans $(1-D,1-D/2)$ and $R$ spans $(0,H(1-D))$.

Based on the above analysis, we obtain the following lemmas and theorems.
\begin{lemma}
\label{lem:Lemmah}Let $h(u)=H(u)-H(u+D-1)$ map $u\in\left[1-D,1-D/2\right)$
to $R$. Then, the inverse mapping of $h$,\vspace{-1mm} \[
h^{-1}:(0,H(1-D)]\to\left[1-D,1-D/2\right),\]
is well-defined and maps $R$ to $u$.\end{lemma}
\begin{proof}
We first notice that $h(u)$ is strictly decreasing since the derivative
is negative over $\left[1-D,1-D/2\right)$, hence the mapping $h:\left[1-D,1-D/2\right)\rightarrow(0,H(1-D)]$
is one-to-one. From the analysis above, one can also see that $h$
is onto.\end{proof}
\begin{thm}
Using mBM-1 with $2^{R}$ decoding attempts where $R\in(0,NH(1-\frac{D}{N})]$,
the maximum rate-distortion exponent that can be achieved is\vspace{-1mm}\begin{equation}
F=N\, D_{KL}\left(h^{-1}\left(R/N\right)||\, p\right).\label{eq:ComputeF}\end{equation}
\vspace{-6mm}\end{thm}
\begin{proof}
First, note that in our context where we have a source sequence $x^{N}$
of $N$ i.i.d. source components, the rate and exponent for each source
component is now $\frac{R}{N}$ and $\frac{F}{N}$. From Case 3 in
the analysis above and from Lemma \ref{lem:Lemmah}, we have\vspace{-1mm}
\[
F/N=D_{KL}(u||p)=D_{KL}\left(h^{-1}\left(R/N\right)||\, p\right)\]
and the theorem follows.\end{proof}
\begin{lemma}
Let $g(u)=D_{KL}(u||p)$ map $u\in[1-D,p]$ to $F$. Then, the inverse
mapping of $g$,\vspace{-1mm} \[
g^{-1}:[0,D_{KL}(1-D\,||\, p)]\rightarrow[1-D,p]\]
is well-defined and maps $F$ to $u$.\end{lemma}
\begin{proof}
We first see that $g(u)$ is a strictly convex function and achieved
minimum value at $u=p$ and therefore $g(u)$ is strictly decreasing
over $[1-D,p]$. Thus, the mapping $g:[1-D,p]\to[0,D_{KL}(1-D\,||\, p)]$
is one-to-one. From the analysis above, one can also see that $g$
is onto.\end{proof}
\begin{thm}
In order to achieve a rate-distortion exponent of $F\in\left[0,N\, D_{KL}\left(1-D\,||\, p\right)\right]$,
the minimum number of decoding attempts required for mBM-1 is $2^{R}$
where\vspace{-1mm}\[
R=N\left[H\left(g^{-1}\left(F/N\right)\right)-H\left(g^{-1}\left(F/N\right)+D/N-1\right)\right]^{+}\]
\end{thm}
\begin{proof}
We also note that the rate, distortion and exponent for each source
component is $\frac{R}{N},\frac{D}{N}$ and $\frac{F}{N}$ respectively.
Combining all the cases in the above analysis, we have \[
R/N=\left[H\left(g^{-1}\left(F/N\right)\right)-H\left(g^{-1}\left(F/N\right)+D/N-1\right)\right]^{+}\]
and the theorem follows. 
\end{proof}

\section{SIMULATION\label{sec:Simulation}}

\begin{figure}[!t]
\begin{minipage}[c][1\totalheight][t]{0.485\textwidth}%
\includegraphics[width=0.4\paperwidth]{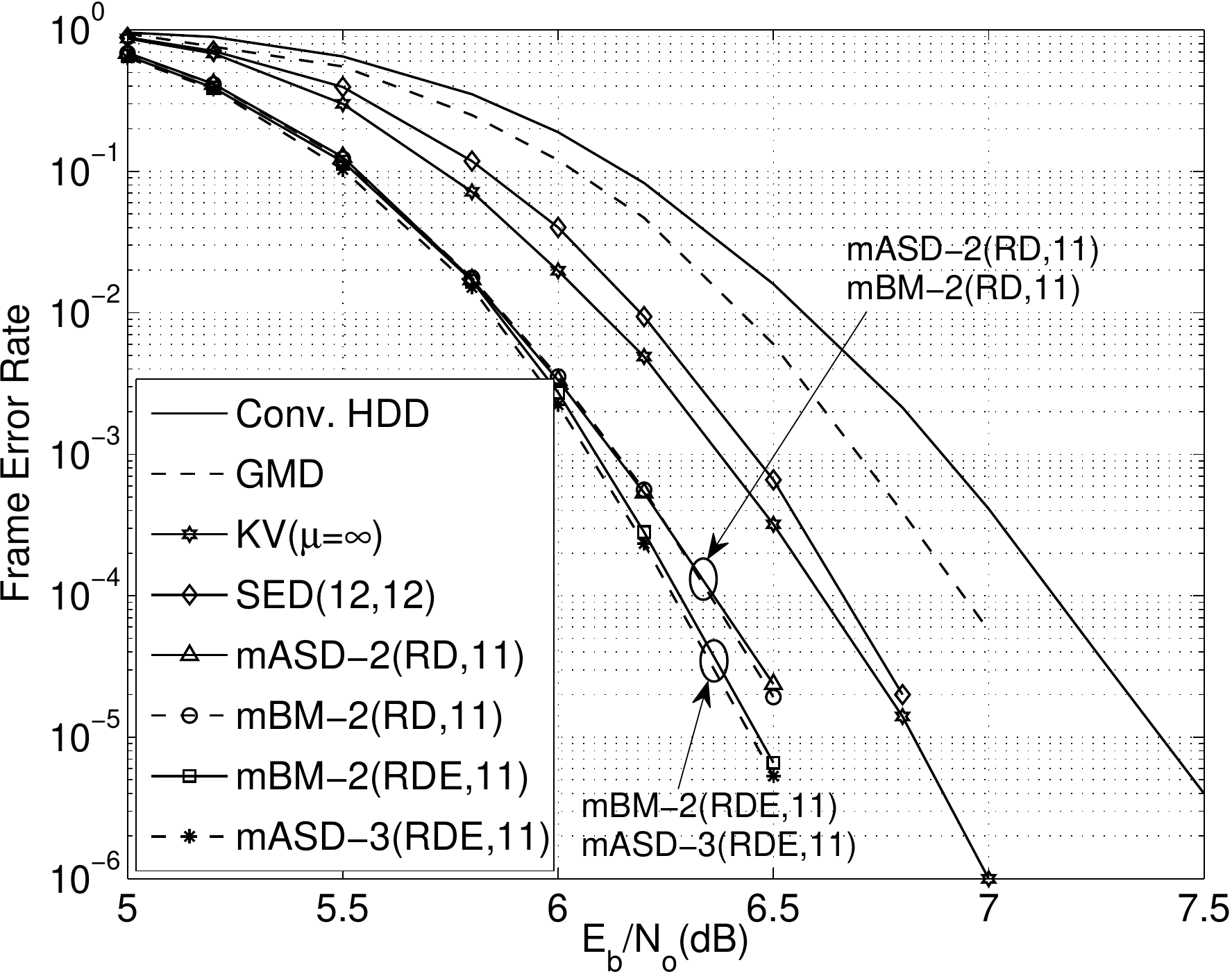}\vspace{-3mm}\caption{\label{fig:simulation} Performance of various decoding algorithms
for the (255,239) RS code over an AWGN channel.}
\end{minipage}
\end{figure}

Simulations of the proposed algorithm were conducted for the (255,239)
RS code over an AWGN channel with BPSK as the modulation format. In
Fig. \ref{fig:simulation}, the mBM-2(RD,11) curve belongs to the
algorithm mBM-2 using RD approach proposed in \cite{Nguyen-aller09}
while the mBM-2(RDE,11) one corresponds to the algorithm mBM-2 using
RDE approach proposed in this paper. The label SED(12,12) denotes
the algorithm presented in \cite{Lee-globecom08}. While all these
three algorithms use the same number of $2^{11}$ erasure patterns,
at a FER of $10^{-4}$, the mBM2(RDE,11) provides a performance gain
of roughly 0.4 dB over the SED(12,12) and outperforms the mBM2(RD,11)
by about 0.1 dB. The conventional HDD and the GMD algorithms have
modest performance since they use only one or a few decoding attempts.
Compared to the conventional HDD, the proposed algorithm mBM-2(RDE,11)
gives approximately a 0.9 dB gain. It also outperforms the Koetter-Vardy
(KV) algorithm \cite{Koetter-it03} with infinite multiplicity $(\mu=\infty)$.
The performance of mBM-2(RDE,11) is roughly the same as the performance
of mASD-3(RDE,11). This implies that, for this setup, algorithms based
on multiple trials of BM decoding perform as good as algorithms based
on running the more complicated ASD decoding multiple times. In Fig.
\ref{fig:simqSC}, we simulate the performance mBM-1(RDE,11) for the
same RS code over an $m$-SC channel. One curve reflects the simulated
frame-error rate (FER) and the other is the approximation derived
from $2^{-F}$ where $F$ is given in (\ref{eq:ComputeF}) with $R=11$.

\begin{figure}[t]
\centering{}\includegraphics[scale=0.4]{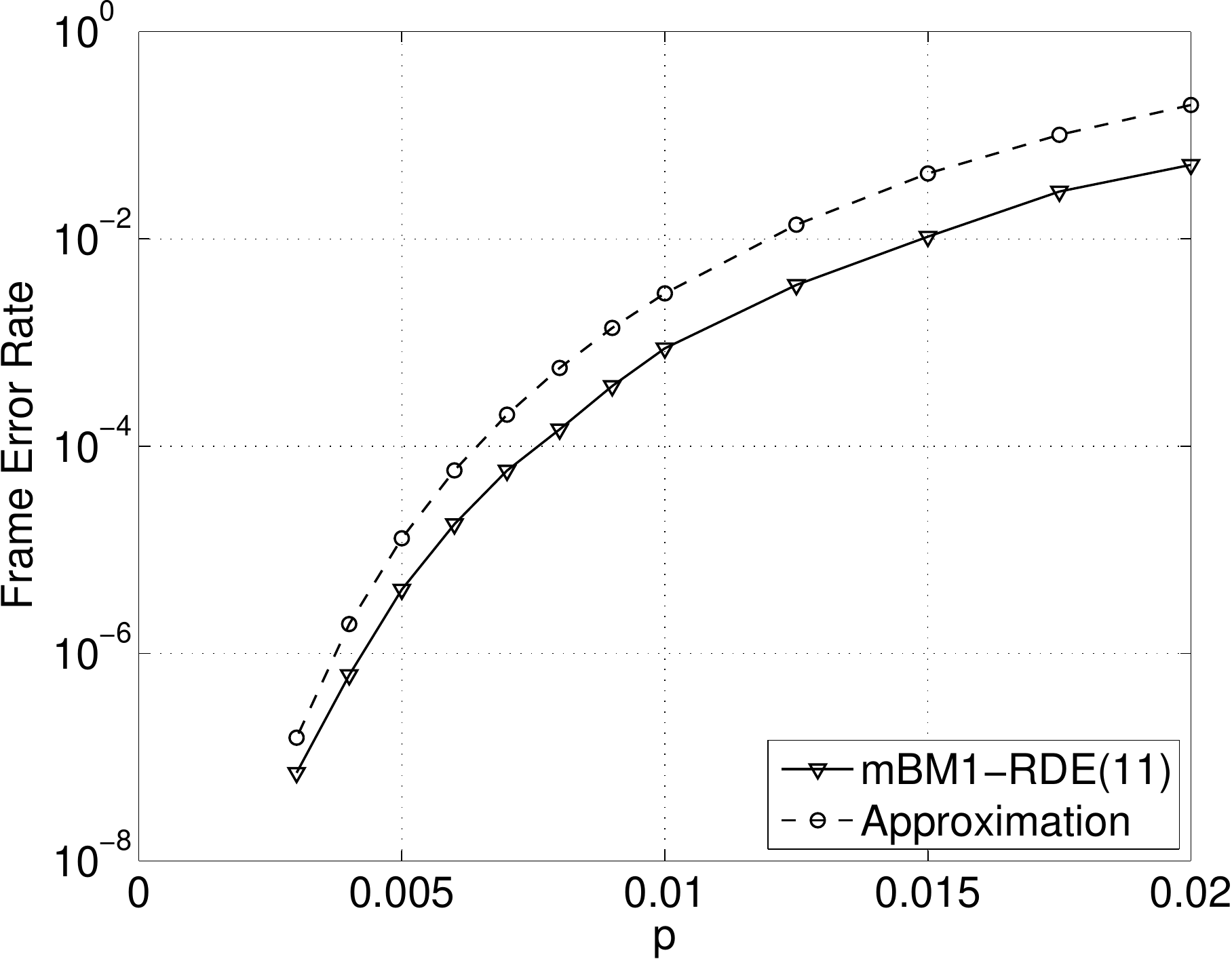}\caption{\label{fig:simqSC} Performance of mBM-1(RDE,11) and its approximation
$2^{-F}$ where $F$ is given in (\ref{eq:ComputeF}) for the (255,239)
RS code over an $m$-SC($p$) channel.}

\end{figure}
\vspace{-1mm}

\section{CONCLUSION\label{sec:Conclusion}}

A RDE-based algorithm has been proposed for multiple decoding attempts
of RS codes. The RDE analysis shows that this approach has several
advantages. Firstly, the RDE approach achieves a near optimal performance-versus-complexity
trade-off among algorithms that consider running a decoding scheme
multiple times (see Remark \ref{rem:The-RDE-approach}). Secondly,
it can help one estimate the error probability using exponentially
tight bounds for $N$ large enough. Simulations are used to confirm
that algorithms using this approach achieve a better trade-off than
previously known algorithms. Along with this, a numerical method is
given to compute the required RDE function.

Our future work focuses on extending this approach to analyze multiple
decoding attempts for ISI channels. In this case, it makes sense for
the decoder to consider multiple candidate error-events during decoding.
Extending the RD approach directly gives a RD problem for Markov sources
in the large distortion regime. Some work is required though because
this is a well-known open problem.

\bibliographystyle{IEEEtran}

\end{document}